\documentclass{article}
\usepackage[utf8]{inputenc}
\usepackage{amsthm}
\usepackage{amssymb}
\usepackage{amsmath}
\usepackage{graphicx}
\usepackage{mathtools}
\usepackage{bbm}
\graphicspath{{./images/}}
\usepackage[margin = 0.5in]{geometry}
\usepackage{mathtools}
\usepackage{pst-node}
\usepackage{tikz-cd}
\usepackage[english]{babel}
\newtheorem{theorem}{Theorem}
\newtheorem{definition}{Definition}

\title{Notes on Lattices, Homomorphic Encryption, and CKKS}
\author{Vir Pathak }
\date{February 2022}

\begin{document}

\maketitle

\section{Introduction}

In cryptography, computing on encrypted data is a major goal which has so far proven elusive. However recent progresses in this area (called homomorphic encryption) have made it possible for such a technology to soon be deployed in the real world. 
\\
\\
In a homomorphic encryption scheme, you have the desirable property that $$\text{Decryption}(\text{Encyption}(a + b)) = \text{Decryption}(\text{Encryption}(a)) + \text{Decryption}(\text{Encryption}(b))$$ and similarly $$\text{Decryption}(\text{Encryption}(ab)) = \text{Decryption}((\text{Encryption}(a))(\text{Encryption}(b)))$$ for real or complex messages $a$ and $b$.
\\
\\
Why is such technology so desirable? Consider the following scenario: imagine you are a hospital which stores vast archives of sensitive medical data and you want an analytics company to train their prediction algorithm on your medical data. Then once trained, you can use such an algorithm to detect diseases in new patients. However, you do not want such sensitive medical information to be leaked to an untrusted third party. The solution is homomorphic encryption, which will allow the company to compute on and train on encrypted data. 
\\
\\
How are people currently trying to create such homomorphic encryption schemes? In 2009, Oded Regev wrote a famous paper reducing the hardness of a linear algebra problem called the Learning With Errors Problem to the quantum hardness of GapSVP. This work won him the 2018 Godel Prize. Moreover, this allows cryptographers to base their schemes off the hardness of LWE, which is extremely versatile and also allows for homomorphic encryption schemes.  
\\
\\
LWE and its Ring variant allowed for cryptographers to devise many partially homomorphic schemes. That is, such schemes do not allow for unbounded addition/multiplication operations on encrypted data. Namely, each operation adds “error” to the encrypted ciphertext. Ultimately, this error grows too large to decrypt properly after enough encrypted computations (particularly multiplication). This is remedied with a technique called bootstrapping which was introduced in Craig Gentry’s PhD thesis in 2009. However, bootstrapping is extremely slow and costly. To bypass bootstrapping, Braskerski, Gentry, and Vaikuntanathan created the BGV scheme. While BGV is technically not fully homomorphic, it is a leveled homomorphic encryption scheme. That is, if one knows the polynomial they want to homomorphically compute beforehand, then they can set the BGV parameters accordingly so that decryption does not fail. 
\\
\\
The above methods make it easy to do homomorphic computations where the encrypted inputs are integers. However, a challenge comes when trying to do approximate arithmetic with real numbers. A naive approach is to scale a real number by a multiplying the input by a large integer (called a scaling factor, say $2^{64}$), do your homomorphic computations with the floor of the scaled number, and finally divide by the integer after the homomorphic computations. However, the problem is that the scaling factor itself increases exponentially after encrypted multiplications. For example, if $c_1$ and $c_2$ are encryptions of $100m_1$ and $100m_2$ respectively (here $100$ is the scaling factor), then homomorphically multiplying $c_1$ and $c_2$ gives $c_3$, which encrypts $10000m_1m_2$. The encryption scheme we talk about in these notes (CKKS) give a method to homomorphically rescale (ie go back to $100m_1m_2$). The security of CKKS comes from the Ring variant of LWE. CKKS also is able to enjoy some useful features by using clever techniques coming from Ring Theory, Linear Algebra, and cyclotomic extensions.  
\\
\\
This is a survey describing what I have learned from the CCS Suf program two years ago as well as what I learned from an REU at Clemson last summer. It will cover lattices, security definitions, and the recent CKKS homomorphic encryption scheme devised by Cheon, Kim, Kim, and Song. In the SUF program two years ago, my main tasks were to read and understand the CKKS paper, implement CKKS, and implement an algorithm optimizing CKKS parameter selection. 
\section{Lattices}
\begin{definition}
Let $v_1,...,v_n \in \mathbb R^m$ be linearly independent vectors. The \textit{lattice} $\mathcal{L}$ generated by $v_1,...,v_n$ is the set of \textit{integer} linear combinations $$\mathcal{L} = \{a_1v_1 +...+a_nv_n \mid a_1,...,a_n \in \mathbb Z\}.$$ The \textit{dimension} of $\mathcal{L}$ is the number of vectors in a basis for $\mathcal{L}$ (so in the above example $\{v_1,...,v_n$\} is the basis and the dimension is $n$). If $n = m$, we say the lattice $\mathcal{L}$ is \textit{full rank}. 
\end{definition}
\subsection{Important Lattice Results}
Let $\{v_1,...,v_n\}$ be a basis for a lattice $\mathcal{L}$ and let $w_1,...,w_m$ be some vectors in $\mathcal{L}$. Then we have 
\begin{align*}
    w_1 &= a_{11}v_1 +...+a_{1n}v_n \\
    & .\\
    & .\\
    & .\\
    w_n &= a_{n1}v_1 +...+a_{nn}v_n
\end{align*}
Now consider the matrix $$A = \begin{pmatrix}
a_{11} & a_{12} & . & . & . a_{1n}\\
.\\
.\\
.\\
a_{n1} & a_{n2} & . &. . & . a_{nn}
\end{pmatrix}.$$ For $v_i$ to be an integer linear combination of the $w_j$, we need $A^{-1}$ to have integer entries. Since $$1 = det(A)det(A^{-1}).$$ and both $det(A)$ and $det(A^{-1})$ are integers, this tells us that the determinant of both matrices are either $1$ or $-1$. Conversely, if $det(A) = \pm 1$ and $A$ has integer entries, then $A^{-1}$ must have integer entries. This gives us the following fact: 
\\
\\
\begin{theorem} Any two bases for a lattice $\mathcal{L}$ are related by a matrix with integer coefficients and a determinant of $\pm 1$. Such matrices are called $\textit{unimodular}$ matrices.
\end{theorem}
A lot of hard lattice problems are centered around finding "short" vectors in the lattice. We can give a naive bound for the shortest vector, namely if we have a lattice basis $B$, then every lattice vector is as least as long as the vectors formed by applying the Gram Schmidt Algorithm to the lattice basis $B$.
\begin{theorem}
Let $B = \{b_1,...,b_n\}$ be a rank $n$ lattice basis and let $\mathcal{L}(B) \subset \mathbb R^m$ be the lattice generated by $B$. Moreover let $\widetilde{B} = \{\widetilde{b_1},...,\widetilde{b_n}\}$ be its Gram-Schmidt orthogonalization. Then $$\lambda_1(\mathcal{L}(B)) \geq \min_{i = 1,...,n}||\widetilde{b_i}|| > 0.$$ Here $\lambda_1(\mathcal{L}(B))$ denotes the length of the shortest nonzero vector in $\mathcal{L}(B)$. 
\end{theorem} 
\begin{proof}
Let $B$ also denote the $m \times n$ matrix formed by writing each $b_i$ as the $i$th column. Then note that we have $$\mathcal{L}(B) = \{Bx \mid x \in \mathbb Z^n\}.$$ Now take any $x \in \mathbb Z^n$. We will show $||Bx|| \geq \min_{i = 1,...,n}||\widetilde{b_i}||.$ For $x = (x_1,...,x_n)$, let $j$ denote the largest number such that $x_j \neq 0$. We have $$\left|\langle{Bx, \widetilde{b_j}}\rangle\right| = \left| \langle{\sum_{i = 1}^{n}x_ib_i, \widetilde{b_j}}\rangle \right|$$ $$= \left| \sum_{i = 1}^{n} x_i\langle{b_i\widetilde{b_j}}\rangle \right| = |x_j|||\widetilde{b_j}||^2$$  However also note that By Cauchy Schwarz we have $$\left| \langle{Bx, \widetilde{b_j}}\rangle \right| \leq ||Bx|| \cdot ||\widetilde{b_j}||.$$ Combining this with the above equality, we have that $$||Bx|| \geq \frac{\left|{\langle{Bx, \widetilde{b_j}}}\rangle\right|}{||\widetilde{b_j}||}$$ $$= |x_j| \cdot ||\widetilde{b_j}|| \geq ||\widetilde{b_j}|| \geq ||\widetilde{b_j}|| \geq \min_{i = 1,...,n}||\widetilde{b_i}|| > 0$$
\end{proof}
This gives a naive lower bound for the length of the shortest lattice vector. However it is enough to show that lattices can be defined independent of a latttice basis. Namely, lattices in $\mathbb R^m$ are nothing but discrete (additive) subgroups of $\mathbb R^m$. Before that, we give another an easy lemma.
\\
\\
\begin{definition} Given $n$ linearly independent vectors $v_1,...,v_n \in \mathbb R^m$, define their \textit{fundamental parallelepiped} to be $$P(v_1,...,v_n):= \{ \sum_{i = 1}^{n}y_iv_i \mid 0 \leq y_i < 1\}.$$
\end{definition}
\begin{theorem}
Let $\mathcal{L}$ be a full rank lattice in $\mathbb R^m$ and let $v_1,...,v_m$ denote linearly independent vectors in $\mathcal{L}$. Then $v_1,...,v_m$ form a lattice basis for $\mathcal{L}$ if and only if $P(v_1,...,v_m) \cap \mathcal{L} = \{0\}$. 
\end{theorem}
\begin{proof}
Let $\{v_1,...,v_m\}$ be a lattice basis for $\mathcal{L}$ and let $v \in P(v_1,...,v_m) \cap \mathcal{L}$. Then since $v_1,...,v_m$ are linearly independent and a lattice, $v$ can be written uniquely as an integer linear combination of $v_1,...,v_m$. But by definition of $P(v_1,...,v_m)$, it must follow that all the scalars in the integer linear combination are zero. So $v = 0$ as claimed.
\\
\\
For the other direction, assume that $\{v_1,...,v_m\} \subset \mathcal{L}$ are such that $P(v_1,...,v_m) \cap \mathcal{L} = 0$. We show the $v_i$ form a lattice basis. Let $v \in \mathcal{L}$ be any lattice vector. Then since the $v_i$ are linearly independent, we know that $$v = \alpha_1v_1 +...+\alpha_mv_m$$ for real $\alpha_i$. But also note that $$\left\lfloor{\alpha_1}\right\rfloor v_1 +...+ \left\lfloor{\alpha_m}\right\rfloor v_m \in \mathcal{L}.$$ Therefore we know that $$v' := (\alpha_1 - \left\lfloor{\alpha_1}\right\rfloor)v_1 +...+(\alpha_m - \left\lfloor{\alpha_m}\right\rfloor) v_m \in \mathcal{L}.$$ But $v'$ is also clearly in $P(v_1,...,v_m)$. So by assumption we have that $v' = 0$. Since the $v_i$ are linearly independent, it must be that $\alpha_i - \left\lfloor{\alpha_i}\right\rfloor = 0$ which means $\alpha_i$ is an integer for all $i$. This implies that the $v_i$ form a lattice basis for $\mathcal{L}$ as claimed. 
\end{proof}
\begin{definition} A subset $\mathcal{L}$ of $\mathbb R^m$ is called a \textit{discrete additive subgroup} if it is an additive subgroup of $\mathbb R^m$ and there exists an $\epsilon > 0$ such that for any vector $v \in \mathcal{L}$, the intersection of sets $$B_{\epsilon}(v) \cap \mathcal{L} = v.$$ 
\end{definition}

\begin{theorem}
A subset $\mathcal{L}$ of $\mathbb R^m$ is a lattice if and only if it is a discrete subgroup of $\mathbb R^m$. 
\end{theorem}
\begin{proof}
First let $\mathcal{L}$ be a lattice in $\mathbb R^m$. We will show it is a discrete additive subgroup= of $\mathbb R^m.$ Clearly $\mathcal{L}$ is an additive subgroup of $\mathbb R^m$. So for all $x,y \in \mathcal{L}$, we know $x - y \in \mathcal{L}$ and therefore $$||x - y|| \geq \lambda_1(\mathcal{L})$$ where $\lambda_1(\mathcal{L})$ is as defined above. But we know from the previous theorem that $$||x - y|| \geq \lambda_1(\mathcal{L}) \geq \min_{i = 1,...,m}||\widetilde{b_i}|| > 0$$ where $B = \{b_1,...,b_m\}$ is a lattice basis and $\widetilde{B} = \{\widetilde{b_1},...,\widetilde{b_m}\}$ is the Gram Schmidt orthogonalization of $B$. Setting $\epsilon = \lambda_1(\mathcal{L})$, we get the first direction.
\\
\\
Now for the other direction. We will construct a lattice basis  $(v_1,...,v_n) \subset \mathbb R^m$ (with $n \leq m$) for our discrete subgroup $\mathcal{L}$. First pick $y \in \mathcal{L}$ such that there is no vector in $\mathcal{L}$ shorter than $y$. Set $y$ equal to $v_1$. Next, we recursively define $v_i$. Assume $v_1,...,v_i$ have been chosen. Choose a vector $y$ in $\mathcal{L}$ which is not in the span of the previous $v_i$. Next, consider the parallelepiped $$P(v_1,...,v_i,y)$$ and note that the parallelepiped contains at least one (but finitely many) points in the $\mathcal{L}$. Choose a vector $$z \in P(v_1,..,v_i) \setminus span(v_1,...,v_i)$$ where we require $z$ to be in $\mathcal{L}$ and $$dist(z,span(v_1,...,v_i))$$ is minimized. Here our notion of distance is given by the distance from $z$ to its orthogonal projection in the subspace $span(v_1,...,v_i)$. Since there are finitely many $z$, such a choice can be made. Set $z = v_{i+1}$. Continue until we cannot construct anymore $v_i$. Note that we can choose at most $m$ $v_i$ since we cannot take more than $m$ linearly independent vectors in $\mathbb R^m$.
\\
\\
Now that we have picked our set $B = (v_1,...,v_n)$, we will show that $\mathcal{L}$ is precisely the $\mathbb Z$ span of $B$. Clearly anything in the $\mathbb Z$ span of $B$ is contained in $\mathcal{L}$. Next, let $v$ be any vector in $\mathcal{L}$ Then by construction of $B$, we know that we can write $$v = \sum_{i = 1}^{n}y_iv_i$$ for $y_i \in \mathbb R$. Take $$v':= \left\lfloor{y_1}\right\rfloor v_1 +...+\left\lfloor{y_n}\right\rfloor v_n.$$ Then $v'$ is clearly an element in $\mathcal{L}$ and so we know $v - v' \in \mathcal{L}$. Consider the Gram-Schmidt orthogonalization of the set $\{v_1,...,v_n\}$ and denote it by $\{\widetilde{v_1},...,\widetilde{v_n}\}$. Then since we can write $$v - v' = (y_n - \left\lfloor{y_n}\right\rfloor)\widetilde{v_n} + s$$ where $s \in span(v_1,...,v_{n-1})$, we have that $$dist(v - v', span(v_1,...,v_{n-1})) = (y_n - \left\lfloor{y_n}\right\rfloor)\cdot ||\widetilde{v_n}||.$$ By the same logic, $$dist(v_n, span(v_1,...,v_{n-1})) = ||\widetilde{v_n}||.$$ But since $(y_n - \left\lfloor{y_n}\right\rfloor) < 1$, we know $$dist(v - v', span(v_1,...,v_{n-1})) < dist(v_m, span(v_1,...,v_{n-1})).$$ But $v_n$ was chosen such that it was closest to $span(v_1,...,v_{n-1})$ so $v - v' \in span(v_1,...,v_{n-1})$. This implies that $y_n - \left\lfloor{y_n}\right\rfloor = 0$. So $y_n$ is an integer. 
\\
\\
Repeating the same argument for each $y_i$ using $\{v_1,...,v_i\}$ tells us that each $y_i$ is an integer. So $\mathcal{L}$ is precisely the $\mathbb Z$-span of $B$ as claimed.
\end{proof}
\subsection{Hard Lattice Problems}
Take a lattice $\mathcal{L} \subset \mathbb R^m$. There are a number of problems which we are believed to be hard for computers to solve efficiently. These problems will be the basis of the homomorphic encryption schemes we will study and are also the basis of many post-quantum encryption schemes.
\\
\\
\textbf{Shortest Vector Problem (SVP):}
In our lattice $\mathcal{L}$ and given a norm $||\cdot||$ (usually we take the Euclidean norm in $\mathbb R^m$), find a vector $v \in L$ such that $||v||$ is shorter than all other vectors in $L$ (note that such a $v$ may not be unique). 
\\
\\
\textbf{Closest Vector Problem (CVP):}
Given a vector $w \in \mathbb R^m$ that is not in $\mathcal{L}$ find a vector $v \in \mathcal{L}$ such that $||w - v||$ is minimized.
\\
\\
CVP is NP hard and SVP is NP hard under a "randomized reduction hypothesis". Here is another problem believed to be reasonably hard for computers to efficiently solve.
\\
\\
\textbf{GapSVP:} Let $\mathcal{L}$ be an $n$ dimensional lattice and let $\psi$ be a real valued function which takes integers as input. Given a lattice basis for $\mathcal{L}$, decide whether $\lambda_1(\mathcal{L}) \leq 1$ or if $\lambda_1(\mathcal{L}) \geq \psi(n)$.  
\\
\\
The security of all the homomorphic encryption schemes we will consider is based on the GapSVP problem. In order to explain why, we need to explain a new problem. 
\section{LWE and RLWE}
Now we will give the \textit{Learning With Errors Problem} and its variants. Intuitively, the LWE problems ask you to solve a "noisy" system of linear equations. 
\\
\\
\textbf{Notation:} For a probability distribution $\chi$, we write $e \leftarrow \chi$ to mean that $e$ is sampled from the distribution $\chi$. If $S$ is a set, we write $s \leftarrow S$ to mean $s$ is sampled from $S$ according to the uniform distribution over $S$.
\\
\\
\subsection{LWE}
Let $q$ be a prime, $\chi$ a probability distribution which outputs "small values" (e.g. the uniform distribution over $[-B,B]$ for $B << \frac{q}{2}$). We say that a sample from $LWE(n,\chi,q)$ is a set of pairs of the form $$\left\{(a_i, \langle{a_i,s\rangle} + e_i) \mid s \leftarrow \mathbb Z_q^n, e_i \leftarrow \chi, a_i \leftarrow \mathbb Z_q^{n}\right\}.$$
\\
\\
The \textit{decision} version of $LWE(n, m, \chi, q)$ is to distinguish $m$ samples from $LWE(n,\chi,q)$ from $m$ uniformly random samples in $\mathbb Z_q^n \times \mathbb Z_q$. The search version of $LWE(n,\chi,q)$ is to recover $s$ given a polynomially bounded amount (in $n\log q)$ of $LWE(n, \chi, q)$ samples. It is known that these two variants are actually equivalent.
\\
\\
In his 2009 paper \textit{On Lattices, Learning With Errors, Random Linear Codes, and Cryptography}, Oded Regev gives the following theorem: 
\\
\\
\begin{theorem} Let $n, p$ be integers and $\alpha \in \{0,1\}$ be such that $\alpha p > 2\sqrt{n}$. If there exists an efficient algorithm that solves $LWE_{p \Psi, \alpha}$, then there exists an efficient quantum algorithm that approximates GapSVP and SIVP (another lattice problem). 
\end{theorem}
That is, if solving GapSVP is hard (for quantum computers - which we believe they are) then solving $LWE$ (both the decision and search versions) is also hard (for classical computers).
\subsection{RLWE}
Fix a prime $q$ and take the ring $$R_{n,q} = \mathbb Z_q[x]/(x^n + 1).$$ An RLWE sample is of the form $(a(x), b(x))$ where $a(x) \in R_{n,q}$ is uniform random and $b(x) = s(x)a(x) + e(x) \in R_{n,q}$. Here $s$ is a fixed secret polynomial and each coefficient $e$ comes from the error distribution $\chi$ we mentioned above. The (search)RLWE problem is to find the polynomial $s(x)$ given a bunch of $RLWE$ samples. The decision RLWE problem is defined analogously as above. We have the following theorem:
\\
\\
\begin{theorem} Suppose that it is hard for polynomial-time quantum algorithms to approximate the search version of the shortest vector problem (SVP) in the worst case on ideal lattices in $R = \mathbb Z[x]/(x^n + 1)$ to within a fixed poly(n) factor. Then any $poly(n)$ number of samples drawn from the R-LWE distribution are pseudorandom to any polynomial-time (possibly quantum) attacker.
\end{theorem}
So the same conclusion we made for LWE also applies for RLWE!
\\
\\
\section{Example: A Public Key Encryption Scheme Based on LWE}
We will now give an example of a (toy) encryption scheme whose security is based on the decision LWE hardness assumption. For this, we first need to give some notions of security and public key encryption.
\\
\\
\begin{definition}
A \textit{public key encryption scheme} is a tuple of algorithms $(Gen, Enc, Dec)$. $Gen$ takes as input a security parameter denoted by $1^n$ for $n$ bit security. $Gen(1^n)$ outputs a pair $(pk, sk)$. For a message $m$ in message space $\mathcal{M}$ (this can be some ambient field $\mathbb Z_p$ for example), $Enc(pk,m)$ outputs a ciphertext $c$, where $c$ is an element of the ciphertext space $\mathcal{C}$ specified by the algorithm. Finally $Dec$ takes in $sk$ and $c$ as inputs and is such that $Dec(sk, c) = m$. 
\end{definition}
Now we need to give notions which will help us define the \textit{security} of a public key encryption scheme.
\\
\\
\begin{definition}
Let $\{X_n\}$ be a sequence of random variables whose support is bitstrings of length polynomial in $n$. We call the sequence $\{X_n\}$ a \textit{probability ensemble}. Now consider two probaility ensembles $\{X_n\}$ and $\{Y_n\}$. For any probabilistic polynomial time algorithm $A$, if we have $$\bigg{|}Pr[s \leftarrow X_n: A(s) = 1] - Pr[s \leftarrow Y_n: A(s) = 1]\bigg{|} \leq \epsilon(n)$$ where $\epsilon(n)$ is a negligible function in $n$, then we say the two distribution ensembles are $\textit{computationally indistinguishable}$. We write $\{X_n\} \approx_c \{Y_n\}$. (Note that formally, the decision LWE problem is naturally stated in terms of computational indistinguishability). 
\end{definition}
The notion of computational indistinguishability is critical to defining several notions for security in cryptosystems. It intuititively tells us that no computationally bounded attacker can tell the difference between two distributions given several samples. In the setting of an encryption scheme, this is a nice property for the distributions of your ciphertexts to have. Namely, we want an encryption scheme to have the property that, a computationally bounded attacker should not be able to tell two different ciphertexts $c_1$ encrypting $m_1$ and $c_2$ encrypting $m_2$ come from different distributions (ie the distribution of ciphertexts encrypting $m_1$ and the distribution of ciphertexts encrypting $m_2$). That is, given ciphertexts, a computationally bounded attacker should learn nothing from them. 
\\
\\
\begin{definition} Consider an encryption scheme specified by the algorithms (Gen, Enc, Dec). The scheme (Gen, Enc, Dec) is said to be \textit{multi message secure} if for all probabilistic polynomial time algorithms $A$ and all polynomials $q(n)$, there exists a negligible function $\epsilon(\cdot)$ such that for all $n \in \mathbb N$ and $m_0,m_1,...,m_{q(n)}, m_0',m_1',...,m_{q(n)}' \in \{0,1\}^n$, $A$ distinguishes between the following distributions with probability at most $\epsilon(n):$ $$\biggr{(}k \leftarrow Gen(1^n): Enc(k, m_i)\biggr{)}_{i = 1}^{q(n)}$$ and $$\biggr{(}k \leftarrow Gen(1^n): Enc(k, m_i')\biggr{)}_{i = 1}^{q(n)}$$  
\end{definition}
Next, we will need the following fact in the security proof of our scheme. This theorem is a special case of the hybrid lemma. It says that computational indistinguishability is a transitive relation between distribution ensembles.
\begin{theorem}
Let $\{X_n\}$, $\{Y_n\}$, and $\{Z_n\}$ be distribution ensembles. If $\{X_n\} \approx_c \{Y_n\}$ and $\{Y_n\} \approx_c \{Z_n\}$, then $\{X_n\} \approx_c \{Z_n\}$.
\end{theorem}
\subsection{Scheme Description and Security Proof}
This scheme only encrypts single bits. Namely, our message space is simply $\{0,1\}$. Let $q$ be a prime, $n,m \in \mathbb Z$, and $\chi$ be a noise distribution where for any $e \leftarrow \chi$ we have that $||e|| \leq q/4m$ with high probability. 
\\
\\
\textbf{Gen($1^n$)}: This algorithm takes in a security parameter $n$ and samples $A \leftarrow \mathbb Z_q^{n \times m}, e \leftarrow \chi^m$ and outputs $(pk, sk)$ where $$sk = s \leftarrow \mathbb Z_q^n, pk = (A, s^TA + e^T).$$ 
\\
\textbf{Enc$(pk, \mu)$}: Denote our $pk$ by $(A, b^T)$ and message $\mu \in \{0,1\}$. We begin by sampling a random bitstring $r \leftarrow \{0,1\}^m$. Output $(Ar, b^T + \mu\cdot \left\lceil{q/2}\right\rceil))$
\\
\\
\textbf{Dec$(sk, (u,v))$}: For secret key $sk$ and ciphertext $(u,v)$ (recall $u$ is a vector of length $n$ and $v$ is a constant), compute $$||v - s^Tu||.$$ If the above quantity is less that $q/4$, output $0$. Otherwise output $1$.
\subsubsection{Correctness}
Consider our ciphertext $(u,v)$. Note that $v - s^Tu = b^Tr + \mu\cdot \left\lceil{q/2}\right\rceil - s^TAr = e^Tr + \mu\left\lceil{q/2}\right\rceil$. If $\mu = 1$ then clearly the decryption correctly outputs $1$. If $\mu = 0$ then by our requirement for $\chi$, we have that $||e^Tr|| \leq m \cdot \frac{q}{4m} = q/4$ with high probability. So decryption correctly outputs $0$ with high probability.
\subsubsection{Security}
To obtain multi message security, it is sufficient to show that for any $k = poly(n)$, we have $$\left((pk, Enc(pk, \mu_1),...,(pk, Enc(pk, \mu_k)) \right) \approx_c \left( (pk, Enc(pk, 0),...,(pk, Enc(pk, 0) \right).$$
By the hybrid lemma, we can obtain the above requirement by just showing $$(pk, Enc(pk, 0)) \approx_c (pk, Enc(pk, 1)).$$ We show this with another hybrid argument. To do that, we construct the following distributions:
\\
\\
\textbf{First Distribution:} $(pk, ct)$ where $$pk = (A, b^T) = (A, S^TA + e^T)$$ for $A \leftarrow \mathbb Z_q^{m \times n}, s \leftarrow \mathbb Z_q^n, e \leftarrow \chi^m$ and $ct = Enc(pk,0)$ $$ = (Ar, b^Tr)$$ for $r \leftarrow \{0,1\}^m$.
\\
\\
\textbf{Second Distribution:} $(pk, ct)$ where $pk = (A, b^T)$ for $A \leftarrow \mathbb Z_q^{m \times n}$ and \textit{random} $b \leftarrow \mathbb Z_q^n$ and $ct = Enc(pk,0)$ $$= (Ar, b^Tr)$$ for $r \leftarrow \{0,1\}^m$.
\\
\\
\textbf{Third Distribution:} $(pk, ct)$ where $pk = (A, b^T)$ for $A \leftarrow \mathbb Z_q^{m \times n}$ and \textit{random} $b \leftarrow \mathbb Z_q^n$ and $ct = (u,v) \leftarrow \mathbb Z_q^n \times \mathbb Z_q$. \\
\\
\textbf{Fourth Distribution:} $(pk, ct)$ where where $pk = (A, b^T)$ for $A \leftarrow \mathbb Z_q^{m \times n}$ and $ct = Enc(pk, 1)$ $$= (Ar, b^Tr + \left\lceil{q/2}\right\rceil)$$ for random $r \leftarrow \{0,1\}^m$.
\\
\\
\textbf{Fifth Distribution:} $(pk, ct)$ where $pk = (A, b^T) = (A, s^TA + e^T)$ for $A \leftarrow \mathbb Z_q^{m \times n}, s \leftarrow \mathbb Z_q^n, e \leftarrow \chi^m$ and $ct = Enc(pk, 1)$ $$= (Ar, b^Tr + \left\lceil{q/2}\right\rceil)$$ for random $r \leftarrow \{0,1\}^m$. We want to show the first distribution is computationally indistinguishable from the fifth  distribution. By the LWE assumption, first distribution is computationally indistinguishable from the second. By a special case of the universal hash lemma, the second is indistinguishable from the third and the third is indistinguishable from the fourth. Finally the LWE assumption implies that that the fourth distribution is computationally indistinguishable from the fifth. By the hybrid argument, we get that the first and fifth distributions are computationally indistinguishable, as required.
\section{Motivation, Notations, and Definitions for CKKS}
Now we use RLWE to construct a homomorphic encryption scheme called CKKS. The great thing about CKKS is that you can perform homomorphic computations on \textit{real/complex numbers}. It was previously only possible to do homomorphic computations on \textit{integers}. This brings us one step closer to real world applications like privacy-preserving machine learning. 
\\
\\
Popular homomorphic encryption schemes have a decryption structure of $\langle{c, sk}\rangle = m + et \pmod{q}$ where $c$ is an encryption of $m$. If we homomorphically multiply encryptions of messages $m_1$ and $m_2$, then the ciphertext corresponding to the encryption of the product $c_1 = m_1m_2$ is contained in some of the least significant bits of the homomorphic multiplication.
\\
\\
This means that if we have encryptions of approximations of $m_1$ and $m_2$, the result $c_1'$ of homomorphically multiplying these encrypted approximations will have completely different Most Significant Bits compared to that of $c_1$. Therefore, it is difficult to homomorphically do arithmetic with encrypted approximations of messages.
\newline
\newline
To do approximate arithmetic homomorphically, we want our decryption to have the our message contained in the most significant bits of of the decryption $\langle{c, sk}\rangle$. Therefore, we want our decryption structure to be of the form $\langle{c, sk}\rangle = m + e$ where $e$ is small compared to $m$.

\subsection{Notation and Definitions}
Let $M$ be a power of $2$ larger where $M > 2$, and let $\phi_M(x)$ be the $M$th cyclotomic polynomial with degree $N = \phi(M)$ (note that $N$ is simply $M/2$ and is always even). More explicitly, $$\phi_M(x) = (x - \zeta)...(x - \zeta^j)...(x - \zeta^{M - 1})$$ where where $\zeta = e^{\frac{2 \pi i}{M}}$ and $j$ runs through all the numbers in $Z_M^{\ast}$. 
\\
\\
Take $R$ to be the ring $\mathbb Z[x]/(\phi_M(x))$. For an integer $q$, we denote $R_q$ to be the quotient ring $R/(qR)$.
\\
\\
Define the canonical embedding $\sigma: R \to \mathbb C^N$ by $\sigma(a) = (a(\zeta),...,a(\zeta^j),...,a(\zeta^{M - 1}))$ where $\zeta = e^{\frac{2 \pi i}{M}}$ and $j$ runs through all the numbers in $Z_M^{\ast}$. Note that $\sigma(a) \in \mathbb C^N$ since $|Z_M^{\ast}| = \phi(M) = N$. For a polynomial $a \in S$, define the canonical embedding norm $$||a||_{\infty}^{\text{can}} = ||\sigma(a)||_{\infty}$$ We use the canonical embedding norm later to discuss the sizes of noise polynomials in decryption and homomorphic operations.
\\
\\
Finally take any real number $a \in \mathbb R$. We write $\lfloor{a}\rceil$ to describe $a$ rounded to the nearest integer.
\section{Encoding/Decoding Procedure}
Our message vectors will be complex vectors in $\mathbb C^{\frac{N}{2}}$ and our plaintext space will be $R$. We do this because we can encrypt multiple messages into one ciphertext. Encoding will map our complex messages to a plaintext polynomial in $R$, while decoding map a plaintext polynomial back to the original complex vector.   
\subsection{Encoding}
Consider the canonical embedding from $R \to \mathbb \sigma(R) \subset \mathbb C^N$. Since $M$ is a power of $2$, we have $N = M/2$. First note that the kernel of this homomorphism is the zero polynomial(denoted by $0$). Since $R/(0)$ is just $R$, we have that $$R \cong \sigma(R)$$ by the first isomorphism theorem. Therefore $\sigma$ maps every polynomial in $R$ to a unique vector in $\sigma(R)$. 
\\
\\
For a vector $z = (z_1,...,z_n) \in \mathbb \sigma(R)$, our goal is to compute $\sigma^{-1}(z)$. This problem reduces to finding the coefficient vector $(\alpha_0,...,\alpha_{N - 1})$ such that $$\sum_{i = 0}^{N - 1}\alpha_i(\zeta^{2j - 1})^i = z_j$$ for $j = 1,...,n$. We can characterize this problem as solving the system of equations $$\begin{bmatrix}
1 && \zeta && \zeta^2&&...&& \zeta^{N - 1}\\
1 && \zeta^3 && (\zeta^3)^2&&...&& (\zeta^3)^{N - 1}\\
&& && && . \\
&& && && . \\
&& && && . \\
1 && \zeta^{2N - 1} && (\zeta^{2N - 1})^2 &&...&& (\zeta^{2N - 1})^{N - 1}\\
\end{bmatrix} \begin{bmatrix}
\alpha_0\\
\alpha_1\\
.\\
.\\
.\\
\alpha^{N - 1}
\end{bmatrix} = \begin{bmatrix}
z_1\\
z_2\\
.\\
.\\
.\\
z_n
\end{bmatrix}$$
\\
\\
Call the large matrix on the left hand side $A$ and the vector on the right hand side $z$. We get our polynomial coefficients by computing $A^{-1}z$. This essentially computes $\sigma^{-1}(z)$.
\newline 
\newline
In the context of CKKS, our plaintext ring is $R = \mathbb Z[x]/(x^N + 1)$. Note that therefore the coefficients of polynomials encoding messages must have integer coefficients. 
\\
\\
Now consider the canonical embedding $\sigma: R \to \mathbb C^N$. Since every $N$th root of unity is the complex conjugate of some other $N$th root of unity, we know that the image $$\sigma(R) \subseteq H = \{z \in \mathbb C^N \mid z_j = \overline{z_{-j}}\}.$$ So we can naturally identify every element of $\sigma(R)$ to be in $\mathbb C^{\frac{N}{2}}$. Therefore, our message space will just be vectors in $\mathbb C^{\frac{N}{2}}$.
\\
\\
Now we give the encoding procedure which sends a message in $\mathbb C^{\frac{N}{2}}$ to a plaintext polynomial in $R$. We start with a message $ z \in \mathbb C^{\frac{N}{2}}$. Then we naturally extend this into an element in $\mathbb C^{N}$ by computing $\pi^{-1}(z)$ which is computed by keeping the $\frac{N}{2}$ entries and then adding the $\frac{N}{2}$ complex conjugates of the first $\frac{N}{2}$ entries to make a vector in $\mathbb H \subset \mathbb C^N$. 
\\
\\
Since $\pi^{-1}(z) \in \mathbb H \subset \mathbb C^{N}$, we would like to immediately apply $\sigma^{-1}$ to $\pi^{-1}$ to get our corresponding plaintext polynomial in $R$. However, note that $\sigma(R) \neq H$, so $\pi^{-1}(z)$ may not be an element of $\sigma(R)$. To fix this problem, we use a technique called coordinate-wise random rounding. We now describe this process. 
\subsubsection{Coordinate-Wise-Random-Rounding and Decoding}
We know that $R$ is isomorphic to $\sigma(R)$. Since $R$ has the orthogonal $\mathbb Z$ basis $(1, x,...,x^{N - 1})$, we know that $$(\beta_1,...,\beta_{N}) = (\sigma(1), \sigma(x),...,\sigma(x^{N - 1}))$$ is an orthogonal $\mathbb Z$ basis for $\sigma(R)$. Then for $z \in \mathbb H$, we simply have to project $z$ onto the basis $(\sigma(1), \sigma(x),...,\sigma(x^{N - 1}))$ to find the closest vector to $z$ in $\sigma(R)$. We have $$z = \sum_{i = 1}^{N}z_i\beta_i$$ with $z_i = \frac{\langle{z, \beta_i}\rangle}{||\beta_i||^2}$. Note $\langle{\cdot}\rangle$ denotes the Hermitian inner product and it turns out that this inner product will always yield real values here.
\\
\\
Then once we compute all the $z_i$, we randomly round(i.e use the coordinate wise random rounding scheme) to round $z_i$ to randomly up or down to the nearest integer. Let $l$ be the vector of randomly rounded $z_i$'s. This is an element of $\sigma(R)$. To manage possible rounding errors, we compute $\Delta \cdot \pi^{-1}(l)$ for some scaling factor $\Delta$. Finally, we calculate our $R$ polynomial $m(x) = \sigma^{-1}(\Delta\cdot \pi^{-1}(l))$ using our previously mentioned reasoning. This concludes the encoding scheme.
\\
\\
Finally, decoding is simple. For a polynomial $m \in R$, compute $\pi \circ \sigma(\left \lfloor {\Delta^{-1} \cdot m}\right \rceil)$. 

\subsection{Example}
We now give a toy example of encoding. We will encode $m = (3+4i, 2 - i) \in \mathbb C^{2}$ with $N = 4$, $M = 8$, and $\Delta = 64$. We see that $$\pi^{-1}(3 +4i, 2 - i) = (3 + 4i, 2 - i, 3 - 4i, 2 + i).$$ Multiplying this vector by $\Delta$, we obtain the vector $z = (192 + 256i, 128 - 64i, 192 - 256i, 128 + 64i)$. We now project $z$ onto $\sigma(R)$ to get a close approximation of $z$ in $\sigma(R)$. By our previous discussion, we have that $z \approx \sum_{i = 1}^{4}z_i\beta_i$ where $$z_i = \frac{\langle{z, \beta_i}\rangle}{||\beta_i||^2}$$ and $\beta_i = \sigma(x^i)$. 
\\
\\
We first compute each $\beta_i$. We see $\beta_1 = \sigma(1) = (1,1,1,1)$. Next, $\beta_2 = (\frac{\sqrt{2}}{2} + \frac{i\sqrt{2}}{2},  -\frac{\sqrt{2}}{2} + \frac{i\sqrt{2}}{2},  \frac{\sqrt{2}}{2} - \frac{i\sqrt{2}}{2},  -\frac{\sqrt{2}}{2} - \frac{i\sqrt{2}}{2})$. Then squaring each entry in $\beta_2$, we have that $\beta_3 = (i, -i, -i, i)$. Finally, $\beta_4 = ( -\frac{\sqrt{2}}{2} + \frac{i\sqrt{2}}{2},  \frac{\sqrt{2}}{2} + \frac{i\sqrt{2}}{2},  -\frac{\sqrt{2}}{2} - \frac{i\sqrt{2}}{2},  \frac{\sqrt{2}}{2} - \frac{i\sqrt{2}}{2})$
\\
\\
With this formula, we get $z_1 = 160$, $z_2 = 90.5$, $z_3 = 160$, $z_4 = 45.2$. Now we randomly round each $z_i$ to the nearest integer to get the vector $z' = (160, 90, 160, 45)$. 
\\
\\
Now we compute our projection. This will be $$v = z_1\beta_1 + z_2\beta_2 + z_3\beta_3 + z_4\beta_4 = \begin{bmatrix}
191.82 + 255.46i\\
128.18 - 64.54i\\
191.82 - 255.46i\\
128.18 - 64.54i
\end{bmatrix}$$ Notice how close an approximation this is to $z$. Now we solve the system $$\begin{bmatrix}
1 & \zeta & \zeta^2 & \zeta^3\\
1 & \zeta^3 & (\zeta^3)^2 & (\zeta^3)^3\\
1 & \zeta^5 & (\zeta^5)^2 & (\zeta^5)^3\\
1 & \zeta^7 & (\zeta^7)^2 & (\zeta^7)^3
\end{bmatrix}\begin{bmatrix}
\alpha_0\\
\alpha_1\\
\alpha_2\\
\alpha_3\\
\end{bmatrix} = \begin{bmatrix}
191.82 + 255.46i\\
128.18 - 64.54i\\
191.82 - 255.46i\\
128.18 - 64.54i
\end{bmatrix}$$ We obtain our coefficients $\alpha_0 = 160, \alpha_1 = 90, \alpha_2 = 160, \alpha_3 = 45$ to get our polynomial $m(x) = 160 + 90x + 160x^2 + 45x^3$.

\section{CKKS Scheme Outline}
We first give a brief description of each function in CKKS and what they are supposed to do. We will shortly describe how each of these functions work. The goal is to construct a leveled fully homomorphic scheme which supports approximations as described before. We begin by fixing integers $p > 0$ and modulus $q_0$. For $0 < l \leq L$ write $q_l = p^l\cdot q_0$. 
\\
\\
We choose an integer $M$ as a function of security parameter $\lambda$. For levels $0 \leq l \leq L$, a ciphertext on level $l$ will be an element of $R_{l}^2$.
\\
\\
Encoding and Decoding are as previously described. Now for plaintexts, we want to describe the algorithms \newline$(KeyGen, Enc, Dec, Add, Mult)$. After decrypting, we decode our plaintext and recover the orignal vector in $\mathbb Z[i]^{\frac{N}{2}}$. As before, we work over the plaintext ring $R = \mathbb Z[x]/(\phi_M(x))$ with $M$ a power of $2$.
\\
\\
We now describe the algorithms mentioned before. We have 
\\
\\
$KeyGen(1^{\lambda}):$ Generates secret key $sk$, public key $pk$, and a public evaluation key $evk$ (this is to do homomorphic multiplications).
\newline
\newline
$Enc_{pk}(m): $ For a polynomial $m \in R$, output a ciphertext $c$ in $R_{q_L}^2$ such that $\langle{c, sk}\rangle = m + e \pmod{q_{L}}$ with the noise $e$ small compared to $m$.
\newline
\newline
$Dec_{sk}(c):$ For ciphertext $c$ at level $l$, compute $m' = m + e = \langle{c, sk}\rangle \pmod{q_l}$ where $m$, $e$, and $c$ are as described in $Enc$.
\newline
\newline
$Add(c_1, c_2):$ For ciphertexts $c_1, c_2$ encrypting $m_1$ and $m_2$, output a ciphertext encrypting $m_1 + m_2$ whose noise is bounded by the sum of the noises associated with $c_1$ and $c_2$.
\newline
\newline
$Mult(c_1, c_2):$ For ciphertexts $c_1, c_2$ encrypting $m_1$ and $m_2$, output a ciphertext encrypting $m_1m_2$ whose noise is bounded by a constant $B_{\text{mult}}$.
\newline
\newline
$Rescale_{l \to l'}(c):$ For a ciphertext $c \in R_{q_l}^2$ at level $l > l'$, output $c' = \left \lfloor {\frac{q_{l'}}{q_l}c} \right \rceil$
\newline
\newline
After many homomorphic multiplications, the size of the message in the resulting ciphertext will have grown exponentially. Through $Rescale$, we are able to reduce the size of this message by a factor of $p$. Rescale also reduces the noise of the homomorphic operations. The idea is as follows. Let $z \in \mathbb C^{\frac{N}{2}}$. We consider the number $pz$ and let $m$ be the polynomial encoding of $pz$. Then $\langle{c, sk}\rangle \approx pz$. If we have two ciphertexts $c_1, c_2$ encrypting $m_1, m_2$, homomorphically multiplying them results in a ciphertext $c'$ such that $\langle{c', sk}\rangle \approx p^2z_1z_2$. We use the rescale function to get a ciphertext encrypting $pz_1z_2$.

\section{CKKS Actual Scheme}
We now give the actual CKKS Scheme in full detail. Let $\sigma > 0$. A sample from $DG(\sigma^2)$ is a vector in $\mathbb Z^N$ whose entries are drawn from a discrete Gaussian with variance $\sigma^2$. For a positive integer $h$, a sample from $HWT(h)$ is a vector in $\{0,\pm 1\}^N$ whose hamming weight is $h$. For $0< \rho < 1$ a sample from $ZO(\rho)$ is a vector in $\{0, \pm 1\}^N$ where the probability of having 1 or $-1$ in entry spot $i$ is $\rho /2$ and $0$ being in entry spot $i$ is $1 - \rho$.
\subsection{Description of the Main Algorithms}
Encoding and Decoding are as previously described.
\\
\\
$KeyGen:$ Given security parameter $\lambda$ and $q_L$, generate integers $M, P$ and a real number $\sigma$. Sample $s \leftarrow HWT(h), a \leftarrow R_{q_L}, e \leftarrow DG(\sigma^2)$. We take $sk = (1, s), pk = (b,a)$ with $b = -as + e \pmod{q_L}$. Sample $a' \leftarrow R_{P\cdot q_L}^2, e' \leftarrow DG(\sigma^2)$. Set $evk = (b', a')$ with $b' = -a's + e' + Ps^2$.
\newline
\newline
$Enc_{pk}(m):$ Sample polynomials $ v\leftarrow ZO(0.5), e_0, e_1 \leftarrow DG(\sigma^2)$ whose coefficients form a vector sampled from $ZO(0.5)$ and $DG(\sigma^2)$ respectively. Compute $c = v \cdot pk + (m+e_0, e_1)$ and reduce to a polynomial in $R_{q_L}$.
\newline
\newline
$Dec_{sk}(c):$For $c \in R_{q_l}^2$ Compute $\langle{c, sk}\rangle \pmod{q_l}$.
\newline
\newline
$Add(c_1, c_2):$ For $c_1, c_2 \in R_{q_l}^2$, compute $c_{\text{add}} = c_1 + c_2 \pmod{q_l}$.
\newline
\newline
$Mult(c_1, c_2):$ For $c_1 = (b_1, a_1), c_2 = (b_2, a_2) \ in R_{q_l}^2$, let $(d_0, d_1, d_2) = (b_1b_2, a_1b_2 + a_2b_1, a_1a_2) \pmod{q_l}$. Output $c_{\text{mult}} = (d_0, d_1) + \left \lfloor {P^{-1}\cdot d_2\cdot evk} \right \rceil \pmod{q_l}$
\newline
\newline
$RS_{l \to l'}(c):$ For $c \in R_{q_l}^2$, output $c' = \left \lfloor {\frac{q_{l'}}{q_l}c} \right \rceil \pmod{q_{l'}}$.
\\
\\
Our encrypted messages have a distribution which is (computationally) indistinguishable from that of an RLWE distribution. This is proven formally using a hybrid argument.
\subsection{Bonus! Homomorphic Message Vector Permutation}
Homomorphically permuting the slots of the original plaintext message can be of critical importance in privacy preserving applications. Galois Theory makes this quite easy to do in practice. 
\\
\\
First note that any element in $\sigma(R)$ is also an element of $\mathbb Q(\zeta_M)$. Now recall that any message vector gets mapped to a vector in $\sigma(R)$ by coordinate-wise random rounding. Moreover, recall that $$\text{Gal}(\mathbb Q(\zeta_M)/\mathbb Q) \cong \mathbb Z_M^{\ast}.$$ Now consider $(z_1,...,z_{\frac{N}{2}})$ to be an element of $\sigma(R)$. Then for any $z_i$ and $z_j$, there exists a map $\varphi_{ij} \in \text{Gal}(\mathbb Q(\zeta_m)/\mathbb Q)$ such that $\varphi_{ij}(z_i) = z_j$. To be explicit, we have that $$\varphi_{ij}(\zeta) = \zeta^{(j^{-1}i)}$$ where $j^{-1} \cdot j \equiv 1 \pmod {M}$. 
\\
\\
Now consider the corresponding plaintext polynomial $m(x)$ (i.e the encoding of $(z_1,...,z_{\frac{N}{2}})$. We can analogously act on $m(x)$ to obtain a polynomial $m'(x)$ by setting $$m'(x) = m(x^{j^{-1}i}).$$ Then one can easily verify that the the decodings of $m(x)$ and $m'(x)$ are just permutations of each other. Moreover, the $j$th element of decoding $m'(x)$ is just the $i$th element of decoding $m(x)$. 
\\
\\
Now consider a ciphertext $c = (c_0, c_1)$ encrypting $m$ (on any level). Consider the vector $$(\varphi_{ij}(c_0), \varphi_{ij}(c_1)).$$ One can show that this decrypts to a polynomial which decodes to the same complex vector as what $m'(x)$ decodes to! This is how we plaintext slots on the corresponding ciphertexts.   
\subsection{Example}
\subsubsection{keygen}
We give an example of key generation, encryption, and decryption. We start with key generation. For this example, we take $p = 5, q_0 = 5$, so $q_0 = 5, q_1 = 4\cdot 5 = 20, q_2 = 80, q_3 = 320, q_4 = 1280$. We continue to take $M = 8, N = 4$, and $R = \mathbb Z[x]/(x^4 + 1)$, and $h = 2$.
\\
\\
We first obtain our secret key $sk$. We have $sk = (1,s)$. To get $s$, we sample a vector in $\{0, \pm 1\}^4$ whose hamming weight is $2$. Here, we select $(0,-1, 1, 0)$. Then we have that $s = -x^2 + x$ and $sk = (1, -x^2 + x)$.
\\
\\
We now sample $a$ from $R_{q_L}$. We choose $a = 103x^3 - 15x^2 + 67x - 221$. To get $e$, we sample a vector from $DG(3.2)$. We choose $(0,0,1,1)$ to obtain $e = x + 1$. To create our public key $pk = (b,a)$, we simply compute $b = -as + e \pmod{q_L}$. This turns out to be $$pk = (82x^3 - 288x^2+119x+119,103x^3 - 15x^2 + 67x - 221)$$
\subsubsection{Encryption}
We will encrypt our original message polynomial $m = 160 + 90x + 160x^2 + 45x^3$. We choose $v = x^3 + 1$, $e_0 = x^3 - 1$, $e_1 = x^2 - 1$. We have $$(m + e_0, e_1) = (46x^3 + 160x^2 + 90x + 159, x^2 - 1).$$ We now compute $v \cdot pk \pmod{q_L}$. We have $$(v\cdot b, v\cdot a) = (201x^3 -370x^2 + 407x, -118x^3 - 118x^2 + 82x - 288).$$ Finally we calculate $$c = v\cdot pk + (m + e_0, e_1) = (247x^3-210x^2+497x+159, -118x^3 - 117x^2 + 82x - 289)$$ 
\subsubsection{Decryption}
We now decrypt. This is simply $c[0] + c[1]\cdot s \pmod{q_L}$. We decrypt to get $$48x^3 + 161x^2 + 90x + 160.$$ Notice how close this is to our original plaintext polynomial.

\subsubsection{Decoding}
For fun, let us decode and see how close we get to the original vector we encoded. Recall that we encoded the vector $z = (3 + 4i, 2 - i)$. Above, we encrypted and decrypted the polynomial which encoded this vector. We now decode $p(x) = 48x^3 + 161x^2 + 90x + 160.$ We simply evaluate $p(\zeta)$ and $p(\zeta^3)$. We get $$(189.70 + 258.58i, 130.302 - 63.419i).$$ We now multiply our resultant vector by the inverse of our scaling factor. We compute $\frac{1}{64}(189.70 + 258.58i, 130.302 - 63.419i) = (2.96 + 4.04i, 2.03 - .9909i)$. Finally we round to the nearest Gaussian integer to recover our original complex vector $(3 + 4i, 2 - i)$. 

\subsection{Correctness Lemma}
We give a lemma proving the (approximate) correctness of encryption and decryption. The proofs for the correctness of the rescaling algorithm, and the homomorphic operations are similar.
\begin{theorem}
Let $c$ be an encryption of $m$. Then $\langle{c, sk}\rangle = m + e$ with $||e||_{\infty}^{\text{can}} < B_{\text{clean}} = 8\sqrt{2} + \sigma N + 6\sigma\sqrt{N} + 16\sigma \sqrt{hN}$. Take $c = Enc(m)$ and $m = Ecd(z)$ for $z \in \mathbb Z[i]^{\frac{N}{2}}$. Then if $\Delta > N + 2B_{\text{clean}}$ then $Dcd(Dec(c)) = z$. 
\end{theorem}

\begin{proof}
We see that $||\langle{c, sk}\rangle - m||_{\infty}^{\text{can}} = ||v\cdot e + e_0 + e_1\cdot s||_{\infty}^{\text{can}}$. By triangle inequality, we see this is less than or equal to $$||v\cdot e||_{\infty}^{\text{can}} + ||e_0||_{\infty}^{\text{can}} + ||e_1\cdot s||_{\infty}^{\text{can}}$$ By some heuristic arguments, we conclude this is less than or equal to(with high probability) $$8\sqrt{2}\sigma N + 6\sigma \sqrt{N} + 16\sigma\sqrt{hN}$$
\\
\\
We prove the second part of the proposition. If $z \in \mathbb C^{\frac{N}{2}}$, an encryption of $m = Ecd(z)$ is also an encryption of $\Delta \sigma^{-1}\circ \pi^{-1}(z)$ with an error bound $B' = B_{\text{clean}} + \frac{N}{2}$. Therefore, we can bound our error polynomial $e$ by $B'$. We consider $Dcd(m + B') = \sigma(\pi(\left\lfloor{\Delta^{-1}(m + B')}\right\rceil$. For this to equal $\sigma(\pi(\left\lfloor{\Delta^{-1}(m)}\right\rceil$ as desired, it must be the case that $\Delta^{-1}B' < \frac{1}{2}$. Substituting $B' = B_{\text{clean}} + \frac{N}{2}$ into the inequality, we get the desired result.
\end{proof}
CKKS was proposed a few years ago. It is the first scheme which allows for approximate homomorphic computation on real/complex messages. In my REU program at Clemson last year, our group proposed another homomorphic encryption scheme. We do not have a rescaling procedure, so we can only compute on integer messages. In order to reduce error, we instead use a technique called \textit{modulus reduction}. Our scheme was inspired by the closely related Braskersky, Gentry, Vaikuntanathan (BGV) scheme. 
\newpage
\section{References}
1) Vinod Vaikuntanathan's notes (lecture 1 and 2) on $\textit{Lattices in Computer Science}$ 
\\
\\
(link: https://people.csail.mit.edu/vinodv/COURSES/CSC2414-F11/index.html)
\\
\\
2) Vinod Vaikuntanathan's notes (lecture 1) on $\textit{Lattices, Learning With Errors, and Post Quantum Cryptography}$ 
\\
\\(link: https://people.csail.mit.edu/vinodv/CS294/) 
\\
\\
3)Rafael Pass, Abhi Shelat: \textit{A Course in Cryptography}
\\
\\
4) Jung Hee Cheon, Andrey Kim, Miran Kim, Yongsoo Song: \textit{Homomorphic Encryption for Arithmetic of Approximate Numbers}
\\
\\
5) Oded Regev: \textit{On Lattices, Learning With Errors, Random Linear Codes, and Cryptography}
\\
\\
6) Vadim Lyubashevsky, Chris Peikert, Oded Regev: \textit{On Ideal Lattices and Learning With Errors Over Rings}
\\
\\
7) Jeffrey Hoffstein, Jill Pipher, Joseph Silverman: \textit{An Introduction to Mathematical Cryptography}
\\
\\
8) Braskerski, Gentry, Vaikuntanathan: \textit{Fully Homomorphic Encryption Without Bootstrapping}
\end{document}